\documentclass[12pt]{article}
\usepackage{amsmath}
\usepackage{amssymb}
\usepackage{amsfonts}
\usepackage[onehalfspacing]{setspace}

\newtheorem{claim}{Claim}

\newtheorem{definition}{Definition}

\newtheorem{proposition}{Proposition}

\newenvironment{proof}[1][Proof]{\noindent\textbf{#1.} }{\ \rule{0.5em}{0.5em}}
\input{tcilatex}
\begin{document}

\title{The Model Selection Curse\thanks{%
This paper is extracted from a significantly longer working paper titled
\textquotedblleft Incentive-Compatible Estimators\textquotedblright\ (Eliaz
and Spiegler (2018)). We thank Susan Athey, Yoav Binyamini, Assaf Cohen,
Rami Atar, Lorens Imhof, Annie Liang, Benny Moldovanu, Ron Peretz and
especially Martin Cripps for helpful conversations. We are also grateful to
seminar and conference audiences at Aarhus, Bocconi, DICE, UCL, Brown, Yale,
BRIQ, Warwick and ESSET, for their useful comments.}}
\author{Kfir Eliaz\thanks{%
School of Economics, Tel-Aviv University and Economics Dept., Columbia
University. E-mail: kfire@post.tau.ac.il. } \ and Ran Spiegler\thanks{%
School of Economics, Tel Aviv University; Department of Economics,
University College London; and CfM. E-mail: rani@post.tau.ac.il.}}
\maketitle

\begin{abstract}
A \textquotedblleft statistician\textquotedblright\ takes an action on
behalf of an agent, based on the agent's self-reported personal data and a
sample involving other people. The action that he takes is an estimated
function of the agent's report. The estimation procedure involves model
selection. We ask the following question: Is truth-telling optimal for the
agent given the statistician's procedure? We analyze this question in the
context of a simple example that highlights the role of model selection. We
suggest that our simple exercise may have implications for the broader issue
of human interaction with \textquotedblleft machine
learning\textquotedblright\ algorithms.\bigskip \bigskip \bigskip \pagebreak
\end{abstract}

\section{Introduction}

In recent years, actions in ever-expanding domains are taken on our behalf
by automated systems that rely on machine-learning tools. Consider the case
of online content provision. A website obtains information about a user's
personal characteristics. Some of these characteristics are actively
provided by the user himself; others are obtained by monitoring his online
navigation history. The website then feeds these characteristics into a
predictive statistical model, which is estimated on a sample consisting of
observations of other users. The estimated model then outputs a prediction
of the user's ideal content. In domains like autonomous driving or medical
decision making, AI systems are mostly confined to issuing recommendations
for a human decision maker. In the future, however, it is possible that
decisions in such domains will be entirely based on machine learning.

How should users interact with such a procedure? In particular, should they
truthfully share personal characteristics with the automatic system? Of
course, in the presence of a conflict of interests between the two parties -
e.g., when the online content provider has a distinct political or
commercial agenda - the user might be better off if he misreports his
characteristics or deletes \textquotedblleft cookies\textquotedblright\ from
his computer. This is a familiar situation of communication under misaligned
preferences, which seems amenable to economists' standard model of strategic
information transmission as a game of incomplete information (with a common
prior).

However, suppose there is no conflict of interests between the two parties -
i.e., the objective behind the machine-learning algorithm is to make the
best prediction of the user's ideal action. But how do such actual systems
perform this prediction task? Consider a very basic textbook tool like LASSO%
\footnote{%
Least Absolute Shrinkage and Selection Operator} (Tibshirani (1996)). This
is a variant on standard linear regression analysis, which adds a cost
function that penalizes non-zero coefficients. The procedure involves both
model selection (i.e. choosing which of many available variables will enter
the regression) and estimation of the selected variables' coefficients. The
predicted action for an agent with a particular vector of personal
characteristics $x$ is the dependent variable's estimated value at $x$. Such
a\ procedure is considered useful when users have many potentially relevant
characteristics relative to the sample size, and especially when we can
expect few of them to be relevant for predicting the agent's ideal action
(i.e., the true data-generating process is \textit{sparse}).

However, LASSO is not a fundamentally Bayesian procedure. Although one can
justify its estimates as properties of a Bayesian posterior derived from
some prior (Tibshirani (1996), Park and Casella (2008), Gao et al. (2015)),
these properties are not necessarily relevant for maximizing the user's
welfare. Furthermore, there is no reason to assume that the prior that
rationalizes LASSO in this manner coincides with the user's actual prior
beliefs (the priors in the above-cited papers involve Laplacian
distributions over parameters). Thus, neither the preferences nor the priors
that take part in the Bayesian foundation for LASSO are necessarily the ones
an economic modeler would like to attribute to the user in a plausible model
of the interaction.

This observation could be extended to many machine-learning predictive
methods that are far more elaborate than the simple textbook example of
penalized regression. If we want to model human interaction with such
algorithms, some departure from the standard Bayesian framework with common
priors seems warranted. Put differently, if one were to analyze a model with
common priors, where a benevolent Bayesian decision maker tries to take the
optimal action for an agent with unknown characteristics, then for almost
all prior beliefs, the decision maker's behavior will not be perfectly
mimicked by a familiar machine-learning procedure. Our approach in this
paper is to take the statistician's procedure as $given$ (rather than trying
to provide a formal rationalization for it) and examine the user's strategic
response to it.

Machine-learning algorithms can be extremely complicated. Nevertheless, in
this paper we follow the tradition of using simple \textquotedblleft
toy\textquotedblright\ models to get insight into complex phenomena.
Economists have developed models in this tradition to study the behavior of
large organizations or the macroeconomy; surely these are more complex than
the most intricate machine-learning algorithm. Accordingly, our model is
perhaps the simplest that can capture the key element we wish to address -
namely, how the element of model selection in machine-learning algorithms
affect users' self-reporting decision.

Specifically, we present a model of an interaction between an
\textquotedblleft agent\textquotedblright\ and a \textquotedblleft
statistician\textquotedblright\ - the latter is a stand-in for an automated
system that obtains personal data from the agent and outputs an action on
his behalf. The agent has a single binary personal characteristic $x$, which
is his private information. The agent has an ideal action, which is a
function of $x$. This function is unknown. The statistician learns about it
by obtaining noisy observations of $other$ agents' ideal actions. This
sample constitutes the statistician's private information. It is $small$,
consisting of $one$ observation for each value of $x$. The statistician
follows a \textquotedblleft penalized regression\textquotedblright\
procedure: the estimated coefficients of his model minimize a combination of
the Residual Sum of Squares and a cost function that combines $L_{0}$ and $%
L_{1}$ penalties of the coefficient of $x$. The procedure's element of model
selection in this simple example consists of the decision whether to admit $x
$ as a predictor of the agent's ideal action.

With one binary characteristic and two sample points, this environment is as
far from \textquotedblleft big data\textquotedblright\ as one could imagine.
Nevertheless, it shares a crucial feature with a typical big-data
predicament that motivates machine-learning methods: the sample size is
roughly the same as the number of potential explanatory variables, such that
an estimation procedure that does not involve selection or shrinkage risks
over-fitting (e.g., see Hastie et al. (2015)). Indeed, an unpenalized
regression would \textit{perfectly} fit the data. As a result, the estimator
would have high variance and its predictive performance could be poor,
relative to an estimator that excludes $x$ or shrinks its coefficient. Thus,
the merit of our simple example is that it manages to capture in a tractable
manner the over-fitting problem.

We pose the following question: Fixing the statistician's procedure and the
agent's prior belief over the true model's parameters, \textit{would the
agent always want to truthfully report his personal characteristics to the
statistician?} When this is the case for all possible priors, we say that
the statistician's procedure (or \textquotedblleft
estimator\textquotedblright ) is \textit{incentive-compatible}. Our analysis
identifies an aspect of the problem that creates a misreporting incentive.
Because the agent's report of $x$ only matters when this variable is
selected by the statistician's procedure, he should only care about the
distribution of the variable's estimated coefficient conditional on the
\textquotedblleft \textit{pivotal event}\textquotedblright\ in which the
variable's coefficient is not zero. One can construct distributions of the
sample noise for which the estimated coefficient conditional on the pivotal
event is so biased that the agent is better off introducing a counter-bias
by misreporting his personal characteristic.

We refer to this effect as the \textquotedblleft \textit{model selection
curse}\textquotedblright . As the term suggests, the logic is reminiscent of
pivotal-reasoning phenomena like the Winner's Curse in auction theory
(Milgrom and Weber (1982)) or the Swing Voter's Curse in the theory of
strategic voting (Feddersen and Pesendorfer (1996)). The model selection
curse does not disappear with large samples: When the noise distribution is
asymmetric, the statistician's procedure can fail incentive-compatibility
even asymptotically. In contrast, we show that when the sample noise is 
\textit{symmetrically} distributed, the estimator is
incentive-compatible.\medskip

\noindent \textit{Related literature}

\noindent Our paper joins a small literature that has begun exploring
incentive issues that emerge in the context of classical-statistics
procedures. Cummings et al. (2015) study agents with privacy concerns who
strategically report their personal data to an analyst who performs a linear
regression. Caragiannis et al. (2016) consider the problem of estimating a
sample mean when the agents who provide the sample observations want to bias
the mean close to their value. Hardt et al. (2016) consider the problem of
designing the most accurate classifier when the input to the classifier is
provided by a strategic agent who faces a cost of lying. Chassang et al.
(2012) argue for a modification of randomized controlled trials when
experimental subjects take unobserved actions that can affect treatment
outcomes. Banerjee et al. (2017) rationalize norms regarding experimental
protocols (especially randomization) by modeling experimenters as
ambiguity-averse decision makers. Spiess (2018) studies the design of
estimation procedures that involve model selection when the statistician and
the social planner have conflicting interests (e.g., when the statistician
has a preference for reporting large effects).

\section{A Model}

An \textit{agent }has a privately known, binary personal characteristic $%
x\in \{0,1\}$. In the context of medical decision making, $x$ can represent
a risk factor (e.g. smoking). In the context of online content provision, it
can indicate whether the agent visited a particular website. A \textit{%
statistician} must take an action $a\in 
\mathbb{R}
$ on the agent's behalf. The agent's payoff from action $a$ is $%
-(a-f(x))^{2} $, where $f(x)\in 
\mathbb{R}
$ is the agent's ideal action as a function of $x$. It will be convenient to
write $f(0)=\beta _{0}$ and $f(1)=\beta _{0}+\beta _{1}$, such that $\beta
_{1}$ captures the effect of $x$ on the agent's ideal action. The parameter
profile $\beta =(\beta _{0},\beta _{1})$ is unknown.

Before taking an action, the statistician privately observes a noisy signal
about $f$. Specifically, for each $x=0,1$, he obtains a \textit{single}
observation $y_{x}=f(x)+\varepsilon _{x}$, where $\varepsilon _{0}$ and $%
\varepsilon _{1}$ are drawn $i.i.d$ from some distribution with zero mean.
Denote $\varepsilon =(\varepsilon _{0},\varepsilon _{1})$. The observations
do not involve the agent himself. We have thus described an environment with
two-sided private information: the agent privately knows $x$, whereas the
statistician has private access to the sample $(y_{0},y_{1})$.

Equipped with the sample $(y_{0},y_{1})$, the statistician follows a
\textquotedblleft penalized regression\textquotedblright\ procedure for
estimating $\beta $. That is, he solves the following minimization problem,%
\begin{equation}
\min_{b_{0},b_{1}}\qquad \sum_{x=0,1}(y_{x}-b_{0}-b_{1}x)^{2}+C(b_{1})
\label{univariate_problem}
\end{equation}%
The first term is the standard Residual Sum of Squares, whereas the second
term is a cost associated with $b_{1}$; the intercept $b_{0}$ entails no
cost. (Of course, given our simple set-up, referring to the procedure as a
\textquotedblleft penalized regression\textquotedblright\ is a bit of an
exaggeration.) The solution to (\ref{univariate_problem}) is denoted $%
b(\varepsilon ,\beta \mathbf{)=(}b_{0}(\varepsilon ,\beta
),b_{1}(\varepsilon ,\beta ))$. We refer to $(b(\varepsilon ,\beta \mathbf{)}%
)_{\varepsilon }$ as the \textit{estimator}. The dependence on $(\varepsilon
,\beta \mathbf{)}$ follows from the fact that the estimator is a function of 
$(y_{0},y_{1})$, which in turn is determined by $(\varepsilon ,\beta \mathbf{%
)}$\textbf{.}

We assume the penalty function%
\begin{equation*}
C(b_{1})=c_{0}\mathbf{1}_{b_{1}\neq 0}+c_{1}|b_{1}|
\end{equation*}%
where $c_{0},c_{1}\geq 0$. This is a linear combination of two common
forms:\ a fixed cost for the mere inclusion of a non-zero coefficient ($L_{0}
$ penalty) and a cost for the magnitude of the coefficient in absolute value
($L_{1}$ penalty, or LASSO).\footnote{%
Adding an $L_{2}$ (Ridge) term $c_{2}(b_{1})^{2}$ would not change any of
the results in the paper.} Assume that when the statistician is indifferent
between including and excluding $x$, he includes it.

In the absence of the penalty $C$, the solution to (\ref{univariate_problem}%
) is $b_{0}=y_{0}$, $b_{1}=y_{1}-y_{0}$, such that the Residual Sum of
Squares is zero. In other words, the estimator \textit{perfectly} fits the
data. As a result, the estimator's predictive performance will tend to be
poor - relative to an estimator that sets $b_{0}=\frac{1}{2}(y_{0}+y_{1})$, $%
b_{1}=0$ - when the true value of $\beta _{1}$ is relatively small.

Having estimated $f$, the statistician receives a report $r\in X$ from the
agent. The statistician then takes the action $a=b_{0}+b_{1}r$. The agent's
expected payoff for a given $\beta $ is therefore%
\begin{equation}
-\mathbb{E}_{\varepsilon }\left[ \left( b_{0}(\varepsilon ,\beta
)+b_{1}(\varepsilon ,\beta )r-\beta _{0}-\beta _{1}x\right) \right]
^{2}\medskip  \label{agentutility}
\end{equation}

\noindent This expression can also be written as%
\begin{equation*}
-\mathbb{E}_{\varepsilon }[\hat{f}(r)-f(x)]^{2}
\end{equation*}%
where $\hat{f}(r)=b_{0}(\varepsilon ,\beta )+b_{1}(\varepsilon ,\beta )r$ is
the estimated model's value at the agent's self-report $r$.

Note that the agent's preferences are given by a quadratic loss function.
This is also a standard criterion for evaluating estimators' predictive
success. Suppose that $r=x$ - i.e., the agent submits a truthful report of
his personal characteristic. Then, the agent's expected payoff coincides
with the estimator's mean squared error.

The following are the key definitions of this paper.$\medskip $

\begin{definition}
The estimator is \textbf{incentive compatible at a given prior belief} over
the true model $f$ (i.e. the parameters $\beta $) if the agent is weakly
better off truthfully reporting his personal characteristic, given his
prior. That is,%
\begin{equation*}
\mathbb{E}_{\beta }\mathbb{E}_{\varepsilon }\left[ \hat{f}(x)-f(x)\right]
^{2}\leq \mathbb{E}_{\beta }\mathbb{E}_{\varepsilon }\left[ \hat{f}(r)-f(x)%
\right] ^{2}
\end{equation*}%
for every $x,r\in \{0,1\}$.\medskip
\end{definition}

In this definition, the expectation operator $\mathbb{E}_{\varepsilon }$ is
taken with respect to the given exogenous distribution over the noise
realization profile. The expectation operator $\mathbb{E}_{\beta }$ is taken
with respect to the agent's prior belief over $\beta $.\medskip

\begin{definition}
The estimator is \textbf{incentive compatible }if it is incentive compatible
at every prior belief. Equivalently,%
\begin{equation}
\mathbb{E}_{\varepsilon }\left[ \hat{f}(x)-f(x)\right] ^{2}\leq \mathbb{E}%
_{\varepsilon }\left[ \hat{f}(r)-f(x)\right] ^{2}  \label{IC_multi}
\end{equation}%
for every true model $f$ and every $x,r\in \{0,1\}$.$\medskip $
\end{definition}

Incentive-compatibility means that the agent is unable to perform better by
misreporting his personal characteristic, \textit{regardless} of his beliefs
over the true model's parameters. How should we interpret this requirement,
given that we do not necessarily want to think of the agent as being
sophisticated enough to think in these terms? One interpretation is that
lack of incentive-compatibility is a purely \textit{normative} statement
about the agent's welfare - namely, given how the statistician takes actions
on the agent's behalf, it would be advisable for the agent to misreport.
Furthermore, there are opportunities for new firms to enter and offer the
agent paid advice for how to manipulate the procedure - in analogy to the
industry of \textquotedblleft search engine optimization\textquotedblright .
Incentive-compatibility theoretically eliminates the need for such an
industry. In the context of online content provision, deviating from $x=1$
to $r=0$ can be interpreted as \textquotedblleft deleting a
cookie\textquotedblright . This deviation is straightforward to implement,
and the agent can check if it leads to better content match in the long run.

The agent's expected payoff function is known to be decomposable into two
terms, one capturing the bias of estimator and another its variance.
Comparing the predictive success of different estimators thus boils down to
trading off the estimators' bias and variance. Incentive-compatibility can
thus be viewed as a collection of bias-variance comparisons between two
estimators: one is the statistician's estimator, and another is an estimator
that applies the statistician's procedure to $r$ rather than $x$. The latter
is not an estimation method that a real-life statistician is likely to
propose, but it arises naturally in our setting.

\section{Analysis}

We first derive a complete characterization of the estimator.\medskip 

\begin{proposition}
\label{prop_estimator}The solution to the statistician's minimization
problem (\ref{univariate_problem}) is as follows:%
\begin{equation}
b_{1}(\varepsilon ,\beta )=\left\{ 
\begin{array}{ccc}
\beta _{1}+\varepsilon _{1}-\varepsilon _{0}-c_{1} & if & \beta
_{1}+\varepsilon _{1}-\varepsilon _{0}-\sqrt{(c_{1})^{2}+2c_{0}}\geq 0 \\ 
\beta _{1}+\varepsilon _{1}-\varepsilon _{0}+c_{1} & if & \beta
_{1}+\varepsilon _{1}-\varepsilon _{0}+\sqrt{(c_{1})^{2}+2c_{0}}\leq 0 \\ 
0 & otherwise & 
\end{array}%
\right.   \label{b(1) single}
\end{equation}%
and%
\begin{equation*}
b_{0}(\varepsilon ,\beta )=\frac{1}{2}\left[ y_{0}+y_{1}-b_{1}(\varepsilon
,\beta )\right] \medskip 
\end{equation*}
\end{proposition}

The proof is mechanical and relegated to the supplementary appendix. Note
that $L_{0}$ penalty leads to model selection without affecting the value of 
$b_{1}$ conditional on being non-zero. The $L_{1}$ penalty term leads to
both shrinkage and selection.

Let us now turn to incentive-compatibility. Two factors create a problem in
this regard: sample noise and model selection. Neither factor is problematic
on its own, as the following pair of observations establishes.\medskip 

\begin{claim}
\label{benchmark_nosamplingnoise}Suppose that $\varepsilon =(0,0)$ with
probability one. Then, the estimator is incentive compatible.
\end{claim}

\begin{proof}
Suppose that $\beta _{1}$ is such that $b_{1}=0$. Then, the agent's report
has no effect on the statistician's action, and incentive-compatibility
holds trivially. Now suppose $\beta _{1}$ is such that $b_{1}>0$. Given the
characterization of $b_{1}$, we must have $\beta _{1}-c_{1}\geq 0$. The
statistician's action as a function of the agent's report is $b_{0}$ if $r=0$
and $b_{0}+b_{1}$ if $r=1$, where%
\begin{eqnarray*}
b_{0} &=&\beta _{0}+\frac{1}{2}\beta _{1}-\frac{1}{2}b_{1}=\beta _{0}+\frac{1%
}{2}\beta _{1}-\frac{1}{2}(\beta _{1}-c_{1}) \\
b_{0}+b_{1} &=&\beta _{0}+\frac{1}{2}\beta _{1}-\frac{1}{2}b_{1}+b_{1}=\beta
_{0}+\frac{1}{2}\beta _{1}+\frac{1}{2}(\beta _{1}-c_{1})
\end{eqnarray*}%
When $x=0$ ($x=1$), the agent's ideal action is $\beta _{0}$ ($\beta
_{0}+\beta _{1}$), and since $\beta _{1}-c_{1}\geq 0$, the action $b_{0}$ ($%
b_{0}+b_{1}$) is closer to the ideal point than $b_{0}+b_{1}$ ($b_{0}$).
Thus, honesty is optimal for the agent. A similar calculation establishes
incentive-compatibility when $b_{1}<0$.\medskip 
\end{proof}

\begin{claim}
\label{Prop OLS single}If $c_{0}=c_{1}=0$, then the estimator is
incentive-compatible.
\end{claim}

\begin{proof}
When $c_{0}=c_{1}=0$, we have $b_{1}=(\beta _{1}+\varepsilon
_{1}-\varepsilon _{0})$. Suppose $x=1$ and the agent contemplates whether to
report $r=0.$ In this case inequality (\ref{IC_multi}) can be simplified into%
\begin{equation*}
\mathbb{E}_{\varepsilon }[(b_{1}(\varepsilon ,\beta
))^{2}+2b_{1}(\varepsilon ,\beta )\cdot (b_{0}(\varepsilon ,\beta )-\beta
_{0}-\beta _{1})]\leq 0
\end{equation*}%
Plugging in the expressions for $b_{0}(\varepsilon ,\beta )$ and $%
b_{1}(\varepsilon ,\beta )\ $given by (\ref{b(1) single}), this inequality
reduces to%
\begin{equation}
\mathbb{E}_{\bar{\varepsilon}_{0},\bar{\varepsilon}_{1}}[-(\beta
_{1})^{2}+2\beta _{1}\varepsilon _{0}+(\varepsilon _{1})^{2}-(\varepsilon
_{0})^{2}]\leq 0  \label{IC x=1 to x=0}
\end{equation}%
This inequality holds for all $\beta _{1}$ because $\varepsilon _{0}$ and $%
\varepsilon _{1}$ are $i.i.d$ with mean zero. An analogous argument shows
that an agent with $x=0$ will not benefit from reporting $r=1$.\medskip 
\end{proof}

Thus, sampling noise and model selection are both necessary to produce
violations of incentive-compatibility in our simple set-up. This finding
should not be taken for granted. First, even in the absence of sampling
noise, the penalty $C$ creates a wedge between the statistician's objective
function and the agent's utility. Therefore, it is not obvious a priori that
this de-facto conflict of interest does not give the agent an incentive to
misreport. Second, as long as the agent's prior over $\beta _{1}$ is not
diffuse, the zero-penalty estimator does not produce actions that maximize
his subjective expected utility. This, too, creates a de-facto conflict of
interests between the two parties, which nevertheless does not give the
agent a sufficient incentive to misreport. One might think that the \textit{%
unbiasedness} of the zero-penalty estimator explains Claim \ref{Prop OLS
single}. However, this intuition is misleading because the agent's utility
function involves a \textit{bias-variance} trade-off. As a result, Claim \ref%
{Prop OLS single} breaks down when the statistician draws different numbers
of observations for $x=0$ and $x=1$: the agent may be willing to experience
a biased action due to misreporting because it will reduce its variance.

We now turn to the case of noisy measurement and positive penalties. The
following example illustrates how incentive-compatibility can fail in this
case.\medskip 

\noindent \textit{An example: Bernoulli noise}

\noindent Assume the following noise distribution. For each $x$:%
\begin{equation*}
\varepsilon _{x}=\left\{ 
\begin{array}{cc}
-1 & \text{with probability }p \\ 
&  \\ 
d=p/(1-p) & \text{with probability }1-p%
\end{array}%
\right. 
\end{equation*}%
where $p>\frac{1}{2}$, such that $d>1$. For simplicity, we consider only the 
$L_{0}$ penalty ($c_{1}=0$).

Consider an agent with $x=1$ who reports $r=0$. This misreporting violates
incentive-compatibility if there is some $\beta _{1}$ for which%
\begin{equation*}
\mathbb{E}_{\varepsilon }\left[ b_{0}(\varepsilon ,\beta )+b_{1}(\varepsilon
,\beta )-\beta _{0}-\beta _{1}\right] ^{2}>\mathbb{E}_{\varepsilon }\left[
b_{0}(\varepsilon ,\beta )-\beta _{0}-\beta _{1}\right] ^{2}
\end{equation*}%
Because the agent's misrepresentation only matters in the \textquotedblleft
pivotal event\textquotedblright\ in which $b_{1}(\varepsilon ,\beta )\neq 0$%
, this inequality can be rewritten as 
\begin{equation}
\mathbb{E}_{\varepsilon _{0},\varepsilon _{1}}[-(\beta _{1})^{2}+2\beta
_{1}\varepsilon _{0}+(\varepsilon _{1})^{2}-(\varepsilon _{0})^{2}\mid
(\beta _{1}+\varepsilon _{1}-\varepsilon _{1})^{2}\geq 2c_{0}]>0
\label{bernoulli_IC}
\end{equation}%
For every $\beta _{1}>0$, we can find a range of values for $c_{0}$ such
that $(\beta _{1}+\varepsilon _{1}-\varepsilon _{0})^{2}\geq 2c_{0}$ only
when $\varepsilon _{1}=d$ and $\varepsilon _{0}=-1.$ In this case, (\ref%
{bernoulli_IC}) is reduced to $\beta _{1}<d-1$.

Therefore, every pair of positive numbers $(\beta _{1},c_{0})$ that
satisfies the inequalities 
\begin{eqnarray*}
-(d+1) &<&\sqrt{2c_{0}}-\beta _{1}<d+1 \\
\beta _{1} &<&d-1
\end{eqnarray*}%
will violate incentive-compatibility. The intuition for this effect is as
follows. The pivotal event $\{\varepsilon \mid b_{1}(\varepsilon ,\beta
)\neq 0\}$ is determined by the difference $\varepsilon _{1}-\varepsilon
_{0}.$ For a range of $(\beta _{1},c_{0})$ values, $\varepsilon
_{1}-\varepsilon _{0}=d+1$ with probability one conditional on the pivotal
event. This produces such a biased estimate of $b_{1}$ that the agent wants
to introduce an offsetting bias in the opposite direction by reporting $x=0.$
$\square $\medskip

The example illustrates a feature we refer to as the \textquotedblleft 
\textit{model selection curse}\textquotedblright , in the spirit of the
\textquotedblleft winner's curse\textquotedblright\ and \textquotedblleft
swing voter's curse\textquotedblright . Like these familiar phenomena, the
model selection curse involves statistical inferences from a
\textquotedblleft pivotal event\textquotedblright . Here, the pivotal event
is the inclusion of an explanatory variable in the statistician's predictive
model. The agent's decision whether to misreport his personal characteristic
is relevant only if the statistician's model includes it. Misreporting will
change the statistician's action by $b_{1}(\varepsilon ,\beta )(r-x)$.
Therefore, the agent only cares about the distribution of $b_{1}(\varepsilon
,\beta )$ conditional on the event $\{\varepsilon \mid b_{1}(\varepsilon
,\beta )\neq 0\}$. This distribution can be so skewed that the agent will
prefer to introduce a counter-bias by misreporting.

A key feature of the above example is the asymmetry in the noise
distribution. Our next result shows that this is a crucial feature:
Symmetric noise ensures incentive-compatibility of the statistician's
procedure. For convenience, we consider the case in which the distribution
of $\varepsilon _{x}$ is described by a well-defined density function. The
result is stated for arbitrary $c_{0},c_{1}\geq 0$.\medskip 

\begin{proposition}
\label{Prop IC 1 variable symmetry}If $\varepsilon _{x}$ is symmetrically
distributed around zero, then the estimator is incentive-compatible.
\end{proposition}

\begin{proof}
Consider the deviation from $x=1$ to $r=0.$ This deviation matters only if $%
b_{1}(\varepsilon ,\beta )\neq 0.$ Incentive-compatibility thus\ requires
the following inequality to hold for all $\beta _{0},\beta _{1}$:%
\begin{equation*}
\mathbb{E}_{\varepsilon _{0},\varepsilon _{1}}[(b_{1}(\varepsilon ,\beta
))^{2}+2b_{1}(\varepsilon ,\beta )(b_{0}(\varepsilon ,\beta )-\beta
_{0}-\beta _{1})\mid b_{1}(\varepsilon ,\beta )\neq 0]\leq 0
\end{equation*}%
Plugging the expression for $b_{0}(\varepsilon )\ $given by (\ref{b(1)
single}), this inequality reduces to%
\begin{equation*}
\mathbb{E}_{\varepsilon _{0},\varepsilon _{1}}[b_{1}(\varepsilon ,\beta
)(-\beta _{1}+\varepsilon _{0}+\varepsilon _{1})\mid b_{1}(\varepsilon
,\beta )\neq 0]\leq 0
\end{equation*}

Fix $b_{1}(\varepsilon ,\beta )$ at some value $b_{1}^{\ast }\neq 0$. Define 
$\mathcal{E}(b^{\ast })=\{(\varepsilon _{0},\varepsilon
_{1}):b_{1}(\varepsilon ,\beta )=b_{1}^{\ast }\}.$ Suppose $\mathcal{E}%
(b^{\ast })$ is non-empty. Then, $(u,v)\in \mathcal{E}(b_{1}^{\ast })$
implies that $(-v,-u)\in \mathcal{E}(b^{\ast })$. This follows immediately
from the fact that $b_{1}(\varepsilon ,\beta )\ $is defined by the
difference $\varepsilon _{1}-\varepsilon _{0}$. Because $\varepsilon _{0}$
and $\varepsilon _{1}$ are $i.i.d$ and symmetrically distributed around
zero, the realizations $(u,v)$ and $(-v,-u)$ have the same probability. This
implies that for any given $b_{1}^{\ast }\neq 0,$%
\begin{equation*}
\mathbb{E}_{\varepsilon _{0},\varepsilon _{1}}[b_{1}(\varepsilon ,\beta
)(\varepsilon _{0}+\varepsilon _{1})|b_{1}(\varepsilon ,\beta )=b_{1}^{\ast
}]=0
\end{equation*}%
Therefore, showing that the deviation from $x=1$ to $r=0$ is unprofitable
reduces to showing that%
\begin{equation*}
\beta _{1}\mathbb{E}_{\varepsilon _{0},\varepsilon _{1}}[b_{1}(\varepsilon
,\beta )\mid b_{1}(\varepsilon ,\beta )\neq 0]\geq 0
\end{equation*}%
which simplifies further to%
\begin{equation*}
\beta _{1}\mathbb{E}_{\varepsilon _{0},\varepsilon _{1}}(b_{1}(\varepsilon
,\beta ))\geq 0
\end{equation*}

Suppose without loss of generality that $\beta _{1}>0$. We will show that $%
\mathbb{E}_{\varepsilon _{0},\varepsilon _{1}}(b_{1}(\varepsilon ,\beta
))\geq 0$. Denote $\Delta =\varepsilon _{1}-\varepsilon _{0}$. Let $G$ and $g
$ denote the $cdf$ and density of $\Delta $. Since $\varepsilon _{0}$ and $%
\varepsilon _{1}$ are symmetrically distributed around zero, $g$ is
symmetric. Denote 
\begin{equation*}
c^{\ast }=\sqrt{(c_{1})^{2}+2c_{0}}
\end{equation*}%
We need to show that%
\begin{equation}
\int_{-\infty }^{-c^{\ast }-\beta _{1}}(\beta _{1}+\Delta +c_{1})g(\Delta
)+\int_{c^{\ast }-\beta _{1}}^{\infty }(\beta _{1}+\Delta -c_{1})g(\Delta
)\geq 0  \label{inequality prop2}
\end{equation}%
Denote $t=\beta _{1}+c^{\ast }$, $s=\beta _{1}-c_{1}$, and observe that
because $c^{\ast }\geq c_{1}\geq 0$, $t+s>0$ and $t-s>0$. By the symmetry of 
$g$, (\ref{inequality prop2}) becomes%
\begin{equation}
=\int_{-\infty }^{-t}(t+\Delta )g(\Delta )+\int_{-s}^{\infty }(s+\Delta
)g(\Delta )=tG(-t)+sG(s)+\int_{s}^{t}\Delta g(\Delta )\geq 0
\label{inequality 2 prop 2}
\end{equation}%
Applying integration by parts and the symmetry of $g$, (\ref{inequality 2
prop 2}) becomes%
\begin{equation*}
\int_{-\infty }^{\infty }\Delta g(\Delta )+\int_{-\infty }^{s}G(\Delta
)-\int_{-\infty }^{-t}G(\Delta )\geq 0
\end{equation*}%
Since $\int_{-\infty }^{\infty }\Delta g(\Delta )=\mathbb{E}_{\varepsilon
_{0},\varepsilon _{1}}(\varepsilon _{1}-\varepsilon _{0})=0$, the inequality
we need to prove reduces to%
\begin{equation*}
\int_{-\infty }^{s}G(\Delta )-\int_{-\infty }^{-t}G(\Delta )\geq 0
\end{equation*}%
which holds because $s>-t$.

An analogous argument shows that deviation from $x=0$ to $r=1$ is
unprofitable.\medskip
\end{proof}

The intuition behind this result is that symmetric noise curbs the model
selection curse: although model selection implies that $b_{1}$ is a biased
estimate of $\beta _{1}$, the bias is too small to give the agent the
incentive to introduce the counter-bias that results from misreporting.

\section{Does the Curse Vanish with Large Samples?}

So far, we focused on a sample with two observations, hence, one may think
that the model selection curse is a small-sample phenomenon. In this section
we show that this need not be the case. Extend our model by assuming that
for each $x=0,1$, the statistician obtains $N$ observations of the form $%
y_{x}^{n}=f(x)+\varepsilon _{x}^{n}$, $n=1,...,N$, where $\varepsilon
_{x}^{n}$ is $i.i.d$ with mean zero across all $x,n$. The statistician's
problem is essentially the same:%
\begin{equation*}
\min_{b_{0},b_{1}}\sum_{x=0,1}\dsum%
\limits_{n=1}^{N}(y_{x}^{n}-b_{0}-b_{1}x_{k}^{n})^{2}+N\left( c_{0}\mathbf{1}%
_{b_{1}\neq 0}+c_{1}|b_{1}|\right) 
\end{equation*}%
The entire model and its analysis are unchanged, except that now $%
\varepsilon =(\varepsilon _{0}^{n},\varepsilon _{1}^{n})_{n=1,...,N}$; and
in the solution for the estimator (\ref{b(1) single}), $\varepsilon _{x}$ is
replaced with the average sample noise $\bar{\varepsilon}_{x}=\frac{1}{N}%
\sum_{i=1}^{n}\varepsilon _{x}^{n}$. Denote $\varepsilon =(\varepsilon
_{x}^{n})_{x=0,1;n=1,...,N}$.

Returning to the \textit{Bernoulli-noise example} from the previous section,
we investigate whether the set of parameters that violate incentive
compatibility vanishes as $N\rightarrow \infty $. We continue to assume $%
c_{1}=0$ and restrict attention to the case of $\beta _{1}>0$ - both are
without loss of generality. Suppose that for every $x=0,1$ and every
observation $n=1,...,N$, $\varepsilon _{x}^{n}$ is independently drawn from
the Bernoulli distribution that assigns probability $p>\frac{1}{2}$ to $-1$
and probability $1-p$ to $d=p/(1-p)$. Let $\bar{\varepsilon}_{x}(N)$ denote
the average noise realization over all the $N$ observations for $x\in \{0,1\}
$. The pivotal event $\{\varepsilon \mid b_{1}(\varepsilon ,\beta )\neq 0\}$
can be written as%
\begin{equation}
\left\{ \varepsilon \mid \bar{\varepsilon}_{1}(N)-\bar{\varepsilon}%
_{0}(N)\notin \left( -\sqrt{2c_{0}}-\beta _{1},\sqrt{2c_{0}}-\beta
_{1}\right) \right\}   \label{pivotal event asymptotic}
\end{equation}%
Our goal is find the set of parameters for which incentive-compatibility is
violated in the $N\rightarrow \infty $ limit.\medskip 

\begin{proposition}
\label{prop_largedeviation}The set of parameters $\beta _{1}>0$ and $c_{0},d$
for which incentive-compatibility is violated in the $N\rightarrow \infty $
limit is given by%
\begin{equation}
\beta _{1}<\frac{c_{0}}{\sqrt{2c_{0}}+\frac{2d}{d-1}}
\label{ICviolation_largeN}
\end{equation}
\end{proposition}

\begin{proof}
We first find the limit distribution over $(\bar{\varepsilon}_{0}(N),\bar{%
\varepsilon}_{1}(N))$, conditional on the event (\ref{pivotal event
asymptotic}). To do this, it helps to combine the two samples $(\varepsilon
_{0}^{1},...,\varepsilon _{0}^{N})$ and $(\varepsilon
_{1}^{1},...,\varepsilon _{1}^{N})$ into one composite sample $(\eta
^{1},....,\eta ^{N})$, such that for every $n$, $\eta ^{n}=(\varepsilon
_{1}^{n},\varepsilon _{0}^{n})$. Thus, $\eta ^{n}$ is drawn $i.i.d$
according to the following distribution $\pi $:%
\begin{eqnarray*}
\pi _{-1,-1} &=&\Pr (-1,-1)=p^{2} \\
\pi _{-1,d} &=&\Pr (-1,d)=p(1-p)=\Pr (d,-1)=\pi _{d,-1} \\
\pi _{d,d} &=&\Pr (d,d)=(1-p)^{2}
\end{eqnarray*}%
Denoting by $s_{i,j}$ the empirical frequency of the realization $(i,j)$ in
this composite sample allows us to redefine the pivotal event in terms of a
subset of empirical frequencies $s=(s_{-1,-1},s_{-1,d},s_{d,-1},s_{d,d})$:%
\begin{equation*}
R^{N}=\left\{ s^{N}\mid (s_{d,-1}-s_{-1,d})\notin \left( \frac{-\sqrt{2c_{0}}%
-\beta _{1}}{d+1},\frac{\sqrt{2c_{0}}-\beta _{1}}{d+1}\right) \right\}
\end{equation*}

For any empirical distribution $s$, let $D(s||\pi )$ the relative entropy of 
$s$ with respect to $\pi $:%
\begin{equation}
D(s||\pi )=\sum_{i,j\in \{-1,d\}}s_{i,j}\ln \left( \frac{s_{i,j}}{\pi _{i,j}}%
\right)   \label{KL}
\end{equation}%
Denote 
\begin{equation*}
\theta _{l}=\frac{-\sqrt{2c_{0}}-\beta _{1}}{d+1}\qquad \theta _{h}=\frac{%
\sqrt{2c_{0}}-\beta _{1}}{d+1}
\end{equation*}%
We will now show that in the $N\rightarrow \infty $ limit, the distribution
over $s^{N}$ conditional on $s^{N}\in R^{N}$ assigns probability one to the
unique $s$ that minimizes $D(s||\pi )$ subject to the constraint $%
s_{d,-1}-s_{-1,d}=\theta _{h}$. Recall that we are restricting attention to
a range of parameters such that $-1<\theta _{l}<\theta _{h}<1$. We can
partition the pivotal event $R^{N}$ into two closed intervals: $[-1,\theta
_{l}]$ and $[\theta _{h},1]$. Because $\beta _{1}>0$, $\left\vert \theta
_{l}\right\vert <\left\vert \theta _{h}\right\vert $.

The relative entropy function $D(s||\pi )$ is strictly convex in $s$ and
attains a unique unconstrained minimum of zero at $s=\pi $. Furthermore,
because $\pi _{-1,d}=\pi _{d,-1}$, $D(s||\pi )$ treats $s_{-1,d}$ and $%
s_{-d,1}$ symmetrically. Therefore, for any $\theta \in \lbrack -1,1]$, the
minimum of $D(s||\pi )$ subject to $s_{-1,d}-s_{-d,1}=\theta $ is equal to
the minimum of $D(s||\pi )$ subject to $s_{d,-1}-s_{-1,d}=\theta $, such
that the minimum of $D(s||\pi )$ subject to $s_{d,-1}-s_{-1,d}=\theta $ is
strictly increasing with $\left\vert \theta \right\vert $. Therefore, the
minimum of $D(s||\pi )$ subject to $s_{d,-1}-s_{-1,d}\in \lbrack \theta
_{h},1]$ is strictly below the minimum of $D(s||\pi )$ subject to $%
s_{d,-1}-s_{-1,d}\in \lbrack -1,\theta _{l}]$. By Sanov's Theorem (see
Theorem 11.4.1 in Cover and Thomas (2006, p. 362)), the probability of the
event $[\theta _{h},1]$ is arbitrarily higher than the probability of the
event $[-1,\theta _{l}]$ as $N\rightarrow \infty $. Therefore, we can take
the pivotal event to be $[\theta _{h},1]$. Furthermore, by the conditional
limit theorem (Theorem 11.6.2 in Cover and Thomas (2006, p. 371)), in the $%
N\rightarrow \infty $ limit, the probability that $s_{d,-1}-s_{-1,d}=\theta
_{h}$ conditional on the event $s_{d,-1}-s_{-1,d}\in \lbrack \theta _{h},1]$
is one.

It follows that the objective function is $D(s||\pi )$ and the constraints
are%
\begin{eqnarray*}
s_{d,-1}-s_{-1,d} &=&\frac{\sqrt{2c_{0}}-\beta _{1}}{d+1} \\
s_{-1,-1}+s_{-1,d}+s_{d,-1}+s_{d,d} &=&1
\end{eqnarray*}%
Writing down the Lagrangian, the first-order conditions with respect to $%
(s_{i,j})$ are ($\lambda _{1}$ and $\lambda _{2}$ are the multipliers of the
first and second constraints):%
\begin{eqnarray*}
1+\ln s_{-1,-1}-\ln p^{2}-\lambda _{2} &=&0 \\
1+\ln s_{d,d}-\ln (1-p)^{2}-\lambda _{2} &=&0 \\
1+\ln s_{d,-1}-\ln p(1-p)-\lambda _{1}-\lambda _{2} &=&0 \\
1+\ln s_{-1,d}-\ln p(1-p)+\lambda _{1}-\lambda _{2} &=&0
\end{eqnarray*}%
These equations imply%
\begin{eqnarray*}
s_{d,-1}s_{-1,d} &=&s_{d,d}s_{-1,-1} \\
s_{-1,-1} &=&d^{2}s_{d,d}
\end{eqnarray*}%
Now, since%
\begin{eqnarray*}
d &=&\frac{p}{1-p} \\
\bar{\varepsilon}_{1} &=&(s_{d,-1}+s_{d,d})(d+1)-1 \\
\bar{\varepsilon}_{0} &=&(s_{-1,d}+s_{d,d})(d+1)-1
\end{eqnarray*}%
we have that in the $N\rightarrow \infty $ limit, the distribution over $%
\varepsilon $ conditional on the pivotal event assigns probability one to%
\begin{eqnarray*}
\bar{\varepsilon}_{0} &=&-\frac{1}{2}(\sqrt{2c_{0}}-\beta _{1})-\frac{d}{d-1}%
+\frac{1}{2}\sqrt{(\sqrt{2c_{0}}-\beta _{1})^{2}+\frac{4d^{2}}{(d-1)^{2}}} \\
\bar{\varepsilon}_{1} &=&\frac{1}{2}(\sqrt{2c_{0}}-\beta _{1})-\frac{d}{d-1}+%
\frac{1}{2}\sqrt{(\sqrt{2c_{0}}-\beta _{1})^{2}+\frac{4d^{2}}{(d-1)^{2}}}
\end{eqnarray*}%
Plugging these values into (\ref{bernoulli_IC})) produces the
result.\medskip 
\end{proof}

Thus, the incentive-compatibility problem in the Bernoulli-noise example
does not vanish when the sample is large. Moreover, the more skewed the
underlying noise distribution and the larger the complexity cost, the larger
the set of prior beliefs for which incentive-compatibility is violated in
the $N\rightarrow \infty $ limit. The reason that large samples do not fix
the incentive-compatibility problem is that the agent's reasoning hinges on
the pivotal event in which the variable is included. Therefore, even if the
estimator's unconditional distribution is asymptotically well-behaved, the
relevant question for incentive-compatibility is whether it is well-behaved 
\textit{conditional on the pivotal event}.

Recall that our original assumption of only two observations captured (in a
highly stylized fashion) the idea that model selection can avert
over-fitting. When we continue to assume a single explanatory variable and
raise $N$, the over-fitting problem is attenuated and the role of model
selection diminishes. Indeed, practitioners of penalized regression adjust
penalty parameters to sample size, such that $c_{0},c_{1}\rightarrow 0$ as $%
N\rightarrow \infty $. The key question is therefore whether the $rate$ by
which $c_{0}$ or $c_{1}$ decrease with $N$ is \textit{fast enough} to
outweigh the model selection curse. To answer this question, one needs to
characterize the condition for incentive-compatibility for arbitrary values
of $N,c_{0},c_{1}$. This is an open question that we leave for future work.

\section{Conclusion}

Interactions between humans and machines that follow statistical procedures
are becoming ubiquitous, giving rise to interesting questions for
economists. Our question is whether human decision makers should act
cooperatively toward a machine that employs a non-Bayesian statistical
procedure that aims at good predictions. We demonstrated, via a toy example,
that the element of model selection in the procedure creates non-trivial
incentive issues.

Our little exercise exposed a major methodological challenge. The standard
economic model of interactive decision making is based on the Bayesian,
common-prior paradigm. However, the actual behavior of machine decision
makers is often hard to reconcile with this paradigm. We addressed this
challenge by examining the agent's response to a $fixed$ statistical
procedure with a given specification of its parameters. One would like to 
\textit{endogenize} these choices. However, given that the procedure is
fundamentally non-Bayesian, capturing this endogenization with a
well-defined ex-ante optimization problem is not obvious. Incorporating
incentive-compatibility as a criterion for selecting prediction methods is
therefore conceptually challenging.

In general, modeling strategic interactions that involve machine learning
requires us to depart from the conventional Bayesian framework, toward an
approach that admits decision makers who act as non-Bayesian statisticians.
Such approaches are familiar to us from the bounded rationality literature
(e.g., Osborne and Rubinstein (1998), Spiegler (2006), Cherry and Salant
(2016) and Liang (2018)). Further study of human-machine interactions is
likely to generate new ideas for modeling interactions that involve
boundedly rational $human$ decision makers.

\noindent {\LARGE Appendix: Proof of Proposition \textbf{\ref{prop_estimator}%
\medskip }}

\noindent Fix the realization of sample noise $\varepsilon $. The
coefficients $b_{0}$ and $b_{1}$ are given by the solution to the
first-order conditions of%
\begin{equation*}
\min_{b_{0},b_{1}}\qquad \sum_{x=0,1}(y_{x}-b_{0}-b_{1}x)^{2}+c_{0}\mathbf{1}%
_{b_{1}\neq 0}+c_{1}|b_{1}|
\end{equation*}%
where the dependence of the coefficients $b_{0}$ and $b_{1}$ on the noise
realization $\varepsilon $ is suppressed for notational ease.

The first-order condition with respect to $b_{0}$ is%
\begin{equation}
(y_{0}-b_{0})+(y_{1}-b_{0}-b_{1})=0  \label{FOC b(0) multi}
\end{equation}%
while the first-order condition with respect to $b_{1}$ when $b_{1}\neq 0$
gives%
\begin{equation}
2(y_{1}-b_{0}-b_{1})=sign(b_{1})c_{1}  \label{FOC b(j) multi}
\end{equation}%
In particular, $2(y_{1}-b_{0}-b_{1})=c_{1}$ when $b_{1}>0$, and $%
2(y_{1}-b_{0}-b_{1})=-c_{1}$ when $b_{1}<0$.

From (\ref{FOC b(0) multi}) we obtain%
\begin{equation*}
b_{0}=\frac{1}{2}(y_{0}+y_{1}-b_{1})
\end{equation*}%
Plugging this into (\ref{FOC b(j) multi}), we obtain the following
characterization of $b_{1}$ conditional on it being non-zero:%
\begin{equation*}
b_{1}=\left\{ 
\begin{array}{ccc}
\beta _{1}+\varepsilon _{1}-\varepsilon _{0}-c_{1} & if & b_{1}>0 \\ 
\beta _{1}+\varepsilon _{1}-\varepsilon _{0}+c_{1} & if & b_{1}<0%
\end{array}%
\right. 
\end{equation*}%
This means in particular that when $\beta _{1}+\varepsilon _{1}-\varepsilon
_{0}\in (-c_{1},c_{1})$, $b_{1}=0$.

To complete the characterization of when $b_{1}\neq 0$, we compute the
difference between the Residual Sum of Squares (RSS) when the coefficient $%
b_{1}$ is admitted and when it is omitted. First,%
\begin{equation*}
RSS(b_{1}\neq 0)=\left( b_{0}-y_{0}\right) ^{2}+\left(
b_{0}+b_{1}-y_{1}\right) ^{2}
\end{equation*}%
where $b_{0}$ and $b_{1}$ are given by (\ref{FOC b(0) multi}) and (\ref{FOC
b(j) multi}). In contrast, when $b_{1}$ is omitted, $b_{0}=\frac{1}{2}%
(y_{0}+y_{1})$, such that%
\begin{equation*}
RSS(b_{1}=0)=\left( \frac{1}{2}y_{0}+\frac{1}{2}y_{1}-y_{0}\right)
^{2}+\left( \frac{1}{2}y_{0}+\frac{1}{2}y_{1}-y_{1}\right) ^{2}=\frac{1}{2}%
(y_{1}-y_{0})^{2}
\end{equation*}%
It follows that%
\begin{eqnarray*}
RSS(b_{1} &=&0)-RSS(b_{1}\neq 0)=\frac{1}{2}(y_{1}-y_{0})^{2}-\left(
b_{0}-y_{0}\right) ^{2}-\left( b_{0}+b_{1}-y_{1}\right) ^{2} \\
&=&b_{1}[y_{1}-y_{0}-\frac{1}{2}b_{1}] \\
&=&[y_{1}-y_{0}-sign(b_{1})c_{1}][y_{1}-y_{0}-\frac{1}{2}%
(y_{1}-y_{0}-sign(b_{1})c_{1})] \\
&=&\frac{1}{2}(y_{1}-y_{0})^{2}-\frac{1}{2}(c_{1})^{2} \\
&=&\frac{1}{2}(\beta _{1}+\varepsilon _{1}-\varepsilon _{0})^{2}-\frac{1}{2}%
(c_{1})^{2}
\end{eqnarray*}%
The condition for $b_{1}\neq 0$ is%
\begin{equation*}
RSS(b_{1}=0)-RSS(b_{1}\neq 0)\geq c_{0}
\end{equation*}%
i.e.%
\begin{equation*}
(\beta _{1}+\varepsilon _{1}-\varepsilon _{0})^{2}\geq (c_{1})^{2}+2c_{0}
\end{equation*}%
This concludes the proof.


\bigskip 

\begin{thebibliography}{99}
\bibitem{} Banerjee, A., S. Chassang, S. Montero and E. Snowberg (2017).
\textquotedblleft A Theory of Experimenters,\textquotedblright\ NBER Working
Paper No. 23867.

\bibitem{} Caragiannis, I, Ariel D. Procaccia and N. Shah (2016):\
\textquotedblleft Truthful Univariate Estimators,\textquotedblright\ \textit{%
Proceedings of the 33rd International Conference on Machine Learning} 
\textbf{48}.

\bibitem{} Chassang, S., P. Miquel and E. Snowberg (2012). \textquotedblleft
Selective trials: A Principal-Agent Approach to Randomized Controlled
Experiments,\textquotedblright\ \textit{American Economic Review} \textbf{102%
}, 1279-1309.

\bibitem{} Cherry, J. and Y. Salant (2006). \textquotedblleft Statistical
Inference in Games,\textquotedblright\ Northwestern University Working Paper.

\bibitem{} Cover, T. and J. Thomas (2006). \textit{Elements of Information
Theory}, Second Edition, Wiley.

\bibitem{} Cummings, R., S. Ioannidis and K. Ligett (2015).
\textquotedblleft Truthful Linear Regression,\textquotedblright\ \textit{%
Conference on Learning Theory}, 448-483.

\bibitem{} Eliaz, K. and R. Spiegler (2018). \textquotedblleft
Incentive-Compatible Estimators,\textquotedblright\ Tel-Aviv University
Working Paper.

\bibitem{} Feddersen, T. and W. Pesendorfer (1996). \textquotedblleft The
Swing Voter's Curse,\textquotedblright\ \textit{American Economic Review} 
\textbf{86}, 408-424.

\bibitem{} Gao, C., van der Vaart, A. and H. Zhou (2015). \textquotedblleft
A General Framework for Bayes Structured Linear Models,\textquotedblright\
arXiv preprint arXiv:1506.02174.

\bibitem{} Hardt, M., N. Megiddo, C. Papadimitriou and J. Wootters (2016):\
\textquotedblleft Strategic Classification,\textquotedblright\ \textit{%
Proceedings of the 2016 ACM Conference on Innovations in Theoretical
Computer Science}, 111-122

\bibitem{} Hastie, T., R. Tibshirani and M. Wainwright (2015). \textit{%
Statistical Learning with Sparsity: the LASSO and Generalizations}, CRC
press.

\bibitem{} Liang, A. (2018):\ \textquotedblleft Games of Incomplete
Information Played by Statisticians,\textquotedblright\ University of
Pennsylvania Working Paper.

\bibitem{} Milgrom, P. and R. Weber (1982). \textquotedblleft A Theory of
Auctions and Competitive Bidding,\textquotedblright\ \textit{Econometrica }%
\textbf{50}, 1089-1122.

\bibitem{} Osborne, M. and A. Rubinstein (1998). \textquotedblleft Games
with Procedurally Rational Players,\textquotedblleft\ \textit{American
Economic Review} 88, 834-847.

\bibitem{} Park, T. and G. Casella (2008). \textquotedblleft The Bayesian
Lasso,\textquotedblright\ \textit{Journal of the American Statistical
Association} \textbf{103}, 681-686.

\bibitem{} Spiegler, R. (2006): \textquotedblleft The Market for
Quacks,\textquotedblright\ \textit{Review of Economic Studies} \textbf{73},
1113-1131.

\bibitem{} Spiess, J. (2018). \textquotedblleft Optimal Estimation when
Researcher and Social Preferences are Misaligned,\textquotedblright\ Harvard
University Working Paper.

\bibitem{} Tibshirani, R. (1996). \textquotedblleft Regression Shrinkage and
Selection via the Lasso,\textquotedblright\ \textit{Journal of the Royal
Statistical Society}, Series B (Methodological), 267-288.\pagebreak \bigskip
\bigskip 
\end{thebibliography}
\end{document}